\documentclass{sn-jnl}
\usepackage{amssymb,amsthm,amsmath,amsfonts}
\usepackage{cite}
\usepackage{mathrsfs}
\usepackage{graphicx}
\usepackage{multirow}%
\usepackage[title]{appendix}%
\usepackage[dvipsnames]{xcolor}%
\usepackage{textcomp}%
\usepackage{manyfoot}%
\usepackage{booktabs}%
\usepackage{algorithm}%
\usepackage{algorithmicx}%
\usepackage{algpseudocode}%
\usepackage{listings}%
\usepackage{natbib}

\theoremstyle{definition}
\newtheorem{theorem}{Theorem}
\newtheorem{lemma}[theorem]{Lemma}

\newtheorem{example}[theorem]{Example}
\newtheorem{proposition}[theorem]{Proposition}

\newtheorem{definition}{Definition}
\newtheorem{remark}[theorem]{Remark}
\newtheorem{informalquestion}[theorem]{Informal Question}
\def\path{\mathrm{path}}
\newcommand{\red}[1]{#1}
\newcommand{\redder}[1]{#1}

\begin{document}

\title[Specieslike clusters]{Specieslike clusters based on identical ancestor points}

\author{Samuel Allen Alexander}

\maketitle

\vspace{-1cm}
\begin{center}
    \emph{Accepted for publication in the Journal of Mathematical Biology}
\end{center}
\vspace{.75cm}

\begin{abstract}
\\
\noindent $\mbox{\textbf{Abstract:}}$
We introduce several axioms which may or may not hold for any given subgraph of the directed graph of all organisms (past, present and future) where edges represent biological parenthood, with the simplifying background assumption that life does not go extinct. We argue these axioms are plausible for species: if one were to define species based purely on genealogical relationships, it would be reasonable to define them in such a way as to satisfy these axioms. The main axiom we introduce, which we call the identical ancestor point axiom, states that for any organism in any species, either the species contains at most finitely many descendants of that organism, or else the species contains at most finitely many non-descendants of that organism. We show that this (together with a convexity axiom) reduces the subjectivity of species, in a technical sense. We call connected sets satisfying these two axioms ``specieslike clusters.'' We consider the question of identifying a set of biologically plausible constraints that would guarantee every organism inhabits a maximal specieslike cluster subject to those constraints. We provide one such set consisting of two constraints and show that no proper subset thereof suffices.
\end{abstract}

\section{Introduction}

Studying a two-parent analog of the Wright-Fisher population model
with constant population size $n\gg 1$,
Chang observed \citep{chang1999recent} that with high probability,
everyone alive at least $1.77\log_2(n)$ generations ago was either a
common ancestor of all present-day individuals or a common non-ancestor of
all present-day individuals. More realistic simulations confirmed that
probably within the past ten thousand years, there was a point at which
every living human was either a common ancestor or a common non-ancestor
of all humans alive today; the most recent such time is called the
identical ancestor point \cite{rohde2004modelling} (implicitly: for
\emph{\red{H}omo sapiens} as of the date of that paper's publication).

Chang suggested, in \cite{chang1999recent}
and \cite{changreply}, that this identical ancestor point could
help us better understand species trees ``with a reality that is not tied to
or derived from gene trees'' but purely from genealogy.
From the idea of identical ancestor points, we derive a certain
axiom (which we call the \emph{identical ancestor point axiom})
that can be imposed upon
sets of organisms: some organism-sets satisfy the axiom
and some do not.
\red{We argue that it is plausible that species satisfy this axiom}.

We show that the identical ancestor point axiom plus a certain weak
convexity axiom already uniquely determine species up to a certain
equivalence relation. \red{Connected sets which satisfy these two axioms,
we call \emph{specieslike clusters}. These unfortunately include some very small
sets (such as singletons), which would not usually make good species candidates
(although they can, in some cases, for example if a parentless organism arises from
spontaneous biogenesis and immediately dies). Thus it is natural to consider
maximal specieslike clusters. Unfortunately, there is no guarantee that every
organism inhabits a maximal specieslike cluster.
This, in turn, leads us to ask what additional reasonable
constraints can be imposed so as to guarantee every
organism inhabits a maximal such cluster. We give one such set of constraints
in Theorem \ref{atomicspeciesexistencethm}.}

Given only the birthdates and parent-child relationships of a set of organisms
(a directed graph where edges are directed from parents to children and every
vertex has a real number birthdate),
we would like to classify those organisms into species. How can we do so?
What abstract properties should species satisfy, and what properties can we
assume about the larger graph? In other words, what
distinguishes such biological graphs and their species?
For example, an organism should never be born earlier than its own parents.
One might be tempted to suggest
that organisms are always the same species as their parents, but this would
have undesirable consequences: if the full biological graph is connected, then
this constraint would imply that every species is empty or contains every
organism.

We consider the above questions under the simplifying assumption that the
biological graph of all organisms is infinite, i.e., that life will never fully
go extinct. Alexander pointed out \cite{alexander2013infinite} that
this simplifying assumption is implicit in certain parts of Darwin's
\emph{Origin of Species} \cite{darwin2004origin} and Hennig's
\emph{Phylogenetic Systematics} \cite{hennig1999phylogenetic}.
The idea of including \emph{future} as well as present and past organisms
in our study of species was made explicit by D.J.\ Kornet
\cite{kornet1993permanent,kornet1995internodons,kornet2005composite}.

The layout of the paper is as follows.
\begin{itemize}
    \item
    In Section \ref{biospheresection} we impose some assumptions on the
    parent-child graph of all living organisms.
    \item
    In Section \ref{speciesaxiomssection} we introduce \red{two} species axioms:
    the \emph{identical ancestor point axiom}
    (motivated by identical ancestor points) \red{and}
    the \emph{convexity axiom}\red{.}
    \red{We define \emph{specieslike clusters} to be connected
    sets of organisms satisfying both axioms.}
    % \item
    % In Section \ref{iapsection} we discuss the identical ancestor point property.
    %We also discuss its relationship to the so-called Knight-Darwin Law.
    \item
    In Section \ref{objectivespeciessection} we show that\red{,
    if we ignore organisms whose descendants go extinct, then
    any two infinite specieslike clusters either have almost no overlap
    with eachother or else have almost complete overlap with eachother.
    In a sense, this reduces the subjective nature
    of species definition.}
    \item
    In Section \ref{genericitysection} we develop a general framework for ensuring
    existence of maximal or minimal organism-sets
    subject to \red{various constraints}.
    \item
    In Section \ref{inclusivespeciessectn} we
    \red{prove that every organism inhabits a maximal specieslike cluster
    subject to two additional constraints which we call the \emph{common ancestor
    property} and the \emph{reflection property}}.
    \item
    In Section \ref{objectionssectn} we address some anticipated objections
    \red{against the identical ancestor point axiom}.
    \item
    In Section \ref{summarysection} we summarize and make concluding remarks.
\end{itemize}

\section{Infinite Biospheres}
\label{biospheresection}

``The infinite exhibits itself in different ways--in time, in the generations of man, and in the division of magnitudes.''---Aristotle

A key informal criterion of species is that members of the same species should be
able to have common descendants. Such hypothetical abilities, however, are not
captured by the raw structure of parent-child relations and birthdates. Yes, if
two vertices share a child, we can say they are capable of having common descendants,
but the converse is not so clear. Given only that two organisms have no common
descendant as of today, we cannot conclude that they will never have a common
descendant. Thus, if we are to define species solely in terms of parent-child
relationships and birthdates, it seems these definitions must depend not only
on past and present parent-child relations, but also on future such
relations.
And indeed, not just finitely far in the future, but infinitely far in the future,
for if all life goes extinct at a certain date, then no organisms alive at that
date can possibly share descendants, which renders species classification trivial.
Thus, as simplifying assumption, we shall assume our biological graphs are infinite.

To be clear, we are not claiming that life in reality will necessarily last
forever. Rather, we are considering the infinitude of future life as a simplifying
assumption in the same way that physicists assume translation symmetry in crystals.
No crystal in the real world is infinite, but it is convenient to act as if crystals
are infinite, and the resulting translation symmetries help us understand physics.
In the same way, even if life cannot continue forever in reality, by assuming that
it does continue forever, we can gain insight into the combinatorial
structure of life. See \cite{alexander2013infinite} for additional justification
for this simplifying assumption. \red{We can see just how casually Darwin
incorporated infinity in his biological intuition from language such as
that passage in his
\emph{Voyage of the Beagle} where he writes ``their numbers were infinite;
for the smallest drop of water which I could remove contained very many,''
\cite{beagle} as if infinitely many co-contemporary organisms were present
in the vicinity of the ship! Clearly he did not literally believe that, but
it illustrates he had no fear of infinitary simplifying assumptions.}

In addition to this infinitary assumption, we will make some additional assumptions.
We assume organisms have real number birthdates,
that parents are born before children,
and that at most finitely many organisms have birthdates earlier than any
given real number. Finally, we assume no organism has infinitely many children.

In the context\footnote{In representing the set of all living organisms
as a directed graph, we follow \cite{dress2010species}.}
of a directed graph $G$, we will use the following terminology.
We will say vertex $v$ is a \emph{parent} of vertex $w$ (and that $w$ is a
\emph{child} of $v$) if there is an edge directed from $v$ to $w$.
We will say $v$ is an \emph{ancestor} of $w$ (and that $w$ is a \emph{descendant}
of $v$) if there is a sequence $v=v_1,\ldots,v_n=w$ (with $n>1$)
such that each $v_i$
is a parent of $v_{i+1}$. By abuse of notation, we will write $v\in G$ to
indicate that $v$ is a vertex of $G$, and we will write $S\subseteq G$ (and call
$S$ a \emph{subset} of $G$) if $S$ is a subset of the set of vertices of $G$
\red{(we consider such an $S$ to also be a graph, with the same edges between its
vertices as in $G$)}.

\begin{definition}
\label{infinitebiospheredefn}
    An \emph{infinite biosphere} is \red{a directed graph $G$} along with
    a birthdate function $t$
    assigning a birthdate $t(v)\in \mathbb R$ to every vertex $v$ of $G$,
    such that:
    \begin{enumerate}
        \item Whenever $v$ is a parent of $w$, then $t(v)<t(w)$.
        \item For every $r\in\mathbb R$, at most finitely many $v\in G$
            have $t(v)<r$.
        \item Every vertex has at most finitely many children.
        \item $G$ is infinite.
    \end{enumerate}
\end{definition}

We will not directly use the following lemma, but we will frequently use
reasoning similar to its proof.

\begin{lemma}
\label{koniglemma}
    (K\"onig's Lemma)
    If $G$ is an infinite biosphere and $v\in G$ has infinitely many descendants,
    then there is an infinite directed path in $G$ starting at $v$.
\end{lemma}

\begin{proof}
    Let $v_1=v$. Since $v_1$ has infinitely many descendants in $G$,
    and since $v_1$ has only finitely many children, some child $v_2$ of $v_1$
    has infinitely many descendants in $G$. Likewise, some child $v_3$ of $v_2$
    has infinitely many descendants in $G$. This process continues forever,
    defining an infinite directed path $v_1,v_2,v_3,\ldots$ in $G$ starting at $v$.
\end{proof}

\section{\red{Two} Species Axioms}
\label{speciesaxiomssection}

In this section we will introduce \red{two} axioms, each of which may or may not
apply to any particular set of organisms.
\red{We propose that it is plausible that species satisfy these axioms.}
The most important, and perhaps
most controversial, which we call the identical ancestor point axiom, is an
axiom derived from the idea of identical ancestor points.

\subsection{The identical ancestor point axiom}

\red{One key aspect of a species, identified by
Hennig \cite{hennig1999phylogenetic}, is that a species should not
\red{split} \red{(Kornet later pointed out \cite{kornet1993permanent}
that a species can
temporarily split, and presumably Hennig was talking about
permanent splits)}. But Hennig did not attempt to formally explain
what is meant by this. The following axiom serves as a strong
formalization of lack of permanent splits (for if a species permanently
split into two pieces, an organism could inhabit one of the two pieces
and have offspring in that piece indefinitely, without ever having
offspring in the other piece).}

\begin{definition}
\label{udadefn}
    A set $S\subseteq G$ of organisms satisfies the \emph{identical ancestor point
    axiom} if the following requirement holds. For every $v\in S$,
    \red{at least} one of the following statements is true:
    \begin{itemize}
        \item All but finitely many members of $S$ are descendants of $v$.
        \item All but finitely many members of $S$ are non-descendants of $v$.
    \end{itemize}
    In other words, $S$ satisfies the identical ancestor point axiom if and
    only if there is no $v\in S$ such that $S$ contains both infinitely many
    descendants of $v$ and infinitely many non-descendants of $v$.
\end{definition}

We consider the identical ancestor point axiom to be biologically plausible
as a species axiom. We will give five informal arguments for this position.
These five arguments are not rigorous, but that is okay, as we are only
attempting to motivate the axiom, not to prove it.
To be clear, the theorems in our paper are true whether or not species really
satisfy the identical ancestor point axiom; they are true even if (as some
authors have claimed) species don't actually exist at all. But we hope the
following five arguments will at least make it plausible that species really
might satisfy the axiom.

\textit{First informal argument.}
This first argument, based on identical ancestor points, is what originally
motivated us to formulate the axiom.
Assuming a \red{two-parent variation of the}
Wright-Fisher population model with population size
$n\gg 1$, Chang showed \cite{chang1999recent} that probably
every individual at least $1.77\log_2(n)$ generations before the current
generation is either a common ancestor of the current generation or a common
non-ancestor of the current generation. For each $k>0$,
let $f(k)<k$ be maximal such that every member of the $f(k)$th generation
is either a common ancestor of the $k$th generation or a common non-ancestor
of the $k$th generation\footnote{The idea of this function $f$ is
hinted at by Hein, who writes: ``Because of the effects of isolation, had
we been living in 1700, say, and tried to work out when our universal and
identical ancestors lived, the answers would have been further back in time
than the answers we obtain now'' \cite{hein2004pedigrees}}.
By Chang's result, probably
$f(k)\geq \lfloor k-1.77\log_2(n)\rfloor$.
The exact formula is of course
not important; the important thing is that
$\lim_{k\to\infty} \lfloor k-1.77\log_2(n)\rfloor=\infty$.
This suggests that $\lim_{k\to\infty}f(k)=\infty$.
If so, then for any particular individual $v$, in, say, generation $k$,
there must be some $k'>k$ such that $f(k')>k$, thus
either $v$ is a common ancestor of
generation $k'$, or $v$ is a common non-ancestor of generation $k'$.
Since only finitely many individuals inhabit generations $<k'$,
the whole population (past, present, and future) cannot contain both
infinitely many descendants of $v$ and infinitely many non-descendants of $v$.
\qed

If the above argument holds for humans,
it has some amusing implications. For example,
Plato reports \cite{alcibiades} that Socrates claimed to be a descendant
of Daedalus (and thus allegedly Zeus) in response to Alcibiades claiming
to be a descendant of Eurysaces (and thus allegedly Zeus).
But \red{if the gods stopped parenting new human children before} the
500 BCE \emph{\red{H}omo sapiens} identical ancestor point,
and if at least one human in Socrates' time was a descendant of Zeus,
then \emph{every} human in Socrates' time was a descendant of Zeus.
Perhaps Socrates was trying to hint at the idea of the identical ancestor point.
That would certainly be more in character for a philosopher known for his
humility, not known for bragging about his pedigree.

The above argument can be extended to species capable of asexual reproduction,
too. Wiuf and Hein conjectured \cite{donnelly1999discussion} a formula similar
to Chang's formula for a variation where individuals randomly have either
one or two parents. See also \cite{linder2009common}.

\textit{Second informal argument.}
The reader might wonder: couldn't we deliberately arrange a failure of
the identical ancestor point axiom? For example, we could launch a colony of
humans into space \red{(having no children on earth at launch-time)},
never again to interact with humans on earth or vice versa.
This would violate the identical ancestor point axiom assuming the earthlings
and the space colonists remained human for all time. But this example perfectly
illustrates the \emph{biological} nature of the axiom. For, in this example,
over millions of years, inevitably the two groups would evolutionarily diverge
from each other. They could not \emph{both} remain human forever. The same
would hold if we built an impassable wall around the equator, or sealed off
an island in the Pacific, or tried to violate the axiom in some similar
way (nevermind the great difficulty\footnote{To quote Socrates again:
``It is neither possible nor
beneficial for one tribe to remain alone, in isolation and unmixed.''
\cite{philebus}}
of strictly enforcing such isolation
for eternity).
\red{In this way, we have an interesting interaction between biology and
mathematics:
attempts to invalidate the identical ancestor point axiom by contriving
mathematical counterexamples are
thwarted by evolution inevitably breaking the species in question apart.}
\qed

\textit{Third informal argument.}
We appeal to Hennig \cite{hennig1999phylogenetic}.
Assume a species $S$ fails to satisfy the identical ancestor point axiom.
This means that there is some
organism $v$ in $S$ such that, no matter how far into the future we look, there are
always descendants of $v$ in $S$ and also non-descendants of $v$ in $S$.
If $S_1$ is the set of descendants of $v$ in $S$ and $S_2$ is the set of
non-descendants of $v$ in $S$, then $S_1$ and $S_2$ therefore define at least a
one-way permanent split in $S$, if not a two-way permanent split (it is possible
for members of $S_2$ to have children in $S_1$, but not vice versa).
See Figure \ref{splitfigure} for an illustration of what we are calling
``one-way'' and ``two-way'' permanent splits.
According to Hennig, \red{splits} induce speciation
\red{(later, Kornet \cite{kornet1993permanent}
pointed out that these splits should be permanent)}.
But it is absurd for speciation to occur in $S$, since $S$ itself
is a species.

\redder{To} be sure, it is a little unclear whether Hennig would include
``one-way'' permanent splits as \red{splits}; the example graphs he depicts
all involve two-way splits\footnote{\red{Hennig
does not give a rigorous mathematical
definition of what it means for a network to have or not have a
split. One could
propose the identical ancestor point axiom as a rigorous
definition of what it means for $S$ to not have a permanent split.}}.
The closest we can find is on p.\ 67, where he says:
``In principle, verification [of a species hypothesis] is possible only by proving
that unimpaired reproductive relations are actually possible between all
individuals with the specified characters.'' It seems this is not the case
for the $S,S_1,S_2$ above, since no member of $S_1$ can have a child in $S_2$.

\begin{figure}
    \begin{center}
        \includegraphics[scale=0.5]{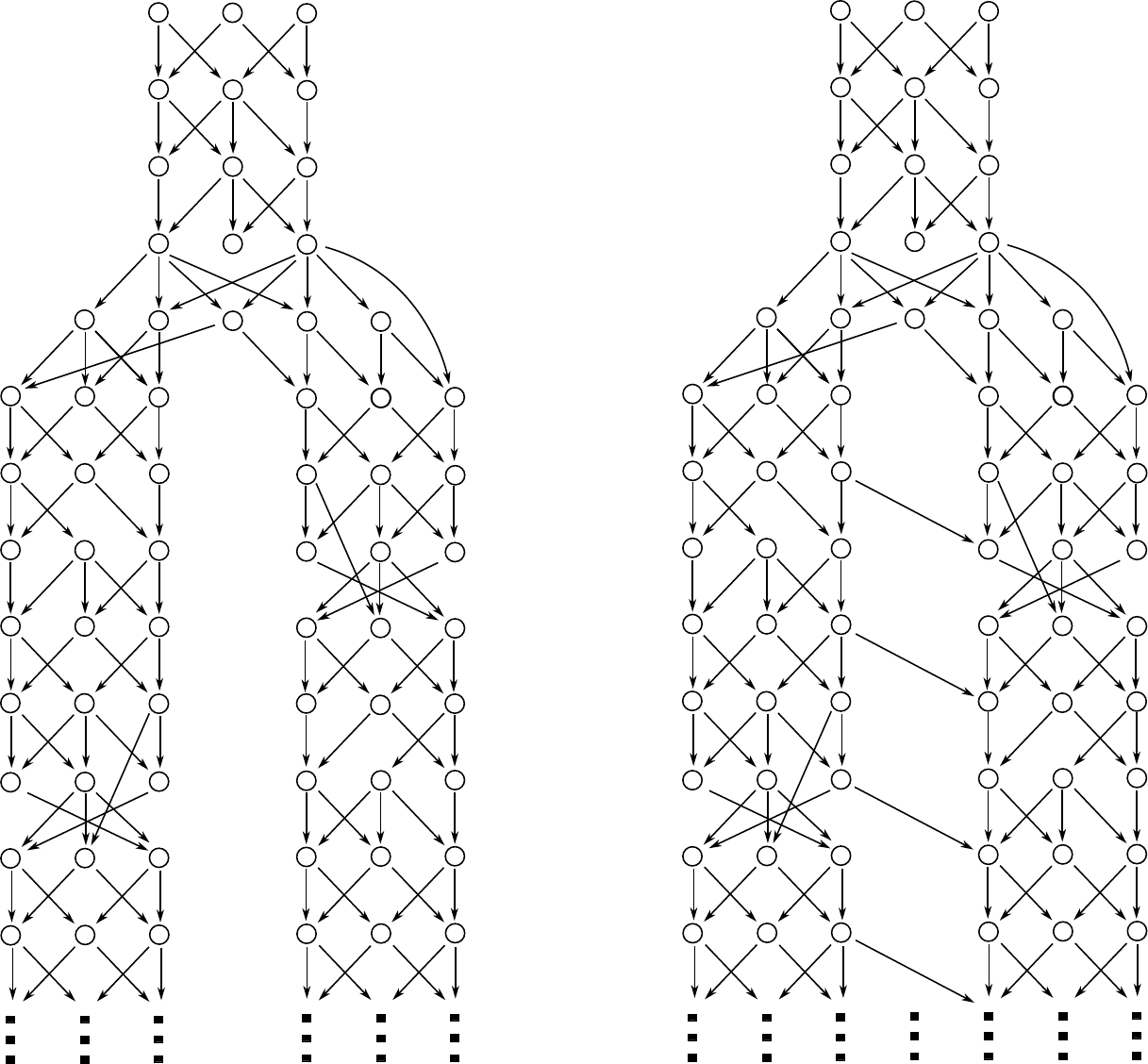}
    \end{center}
    \caption{Different types of permanent splits. Left: A two-way permanent
    split. Right: A one-way permanent split.}
    \label{splitfigure}
\end{figure}

The goatsbeard plant \emph{Tragopogon miscellus},
a hybrid species of other \emph{Tragopogon} species, may be an example of such a
one-way permanent split. Campbell et al write: ``Although the \emph{T.\ miscellus}
population grows mainly by reproduction of its own members,
additional episodes of hybridization between the parent species continue to
add new members to the \emph{T.\ miscellus} population'' \cite{reece2014campbell}.
While of course Campbell et al are not using the Hennigian species concept,
nevertheless the widespread recognition of \emph{T.\ miscellus} as a distinct
species from its parents is evidence that one-way permanent splits
induce speciation.
\qed

\textit{Fourth informal argument.}
We present a simple probability argument.
Let $v$ be some organism in a species, and consider the species some
time after $v$. Suppose that at this time, $X\%$ of the species
consists of descendants of $v$ and $(100-X)\%$ of the species consists
of non-descendants of $v$. If two descendants of $v$ have a child, the child
will be a descendant of $v$; if two non-descendants of $v$ have a child,
the child will be a non-descendant of $v$. So far, the ``competition''
between the $v$-descendant and $v$-non-descendant ``teams'' is fair.
But, if a descendant of $v$ and a non-descendant of $v$ have a child together,
the child will be a descendant of $v$. This asymmetry
gives the $v$-descendants ``team''
an unfair advantage. As a result, we should expect that on average, the
proportion of $v$-descendants will increase until it reaches $100\%$, unless
this team goes extinct before then. Thus, we expect that eventually all
members of the species will be descendants of $v$, or eventually no members
of the species will be descendants of $v$.
\qed

\textit{Fifth informal argument.}
We present an informal combinatorial argument.
Let $v$ be an organism in a species in generation $n$, and consider an
organism $w$ in the same species in generation $n+m$. If $m=1$, then there
are two ``paths'' back in time from $w$ to generation $n$: $w$ to $w$'s mother,
and $w$ to $w$'s father. If $m=2$, then there are four such paths:
\begin{enumerate}
    \item
    $w$ $\rightarrow$ $w$'s mother $\rightarrow$ $w$'s mother's mother.
    \item
    $w$ $\rightarrow$ $w$'s mother $\rightarrow$ $w$'s mother's father.
    \item
    $w$ $\rightarrow$ $w$'s father $\rightarrow$ $w$'s father's mother.
    \item
    $w$ $\rightarrow$ $w$'s father $\rightarrow$ $w$'s father's father.
\end{enumerate}
In general, there are $2^m$ such paths. Assume $m$ is very large, so
$2^m$ is astronomically larger than the size of the
whole generation $n$. Now, if a dartboard has $X$ regions, and we throw
$Y$ darts at it where $Y$ is astronomically larger than $X$, we should expect
every single region to get hit, except possibly some regions which are
unhittable. The regions of the dartboard are a metaphor for the members of
generation $n$; the dart-throws are a metaphor for the $2^m$ paths;
and unhittable regions are a metaphor for members of generation $n$ who have
no living descendants in generation $n+m$. So either $v$ is unhittable,
in which case it has no descendants by generation $n+m$, or else $v$ is hittable,
in which case one of the $2^m$ paths hits it, so $w$ is $v$'s descendant. By
arbitariness of $w$, if $v$ is hittable, $v$ is an ancestor of all members
of generation $n+m$.
\qed

Some of the above arguments assume sexual reproduction.
But they still work in modified form if we merely assume occasional
sexual reproduction. Darwin himself asserted \cite{darwin2004origin}
that occasional sexual reproduction is unavoidable. In fact, he articulated
this in a precise and combinatorial way, in what was later known as
the \emph{Knight-Darwin Law} \cite{darwin1898knight}.
See \cite{alexander2013infinite} for a translation of the Knight-Darwin
Law into modern graph-theoretical language; see \cite{alexander2015alternative}
for precise theorems about how the Knight-Darwin Law sharply constrains
the shape of the asexual parts of the biosphere.
Of course, there is now evidence for some species being genuinely purely
asexual \cite{judson1996ancient}. But even this has been challenged recently.
The most extreme cases are certain species of bdelloid rotifers.
One of those, \emph{Adineta vaga}, might not be asexual after all,
according to a recent paper \cite{vakhrusheva2020genomic}, although the
data is still insufficient to distinguish between meiosis and horizontal
gene transfer. And if it is just horizontal gene transfer, we could still
follow Dress et al \cite{dress2010species}, who consider
horizontal gene transfer to qualify as
a parent-child relationship in their version of $G$. 

Finally, we emphasize that in the identical ancestor point axiom, we are
talking about genealogical ancestorhood, not genetic ancestorhood.
It is possible for an organism to be an ancestor of almost every member
of its species and yet be a genetic ancestor of almost no members of its
species. In fact, Gravel and Steel have argued that this happens with high
probability \cite{gravel2015existence}.
\red{In Section \ref{objectionssectn} we address several anticipated objections
against the identical ancestor point axiom; readers can safely read
ahead to that section if they so desire.}

\subsection{The convexity axiom}

\begin{definition}
\label{convexitydefn}
    A set $S\subseteq G$ of organisms satisfies the \emph{convexity axiom}
    if the following requirement holds. For every $v\in G$,
    if $v$ has an ancestor in $S$ and also $v$ has a descendant in $S$,
    then $v\in S$.
\end{definition}

Convexity seems biologically plausible
\red{for species}. If $v_1$ is an ancestor of $v_2$ and
$v_1$ is in a species $S$, this might not necessarily imply $v_2\in S$,
because maybe $v_2$ (or an intermediate organism between $v_1$ and $v_2$)
mutated into a new species. But if additionally $v_2$ is an ancestor of $v_3\in S$,
then it seems unlikely that the line from $v_1$ to $v_3$ could have mutated into
a new species by the time it reached $v_2$, and then mutated back to the original
species again by the time it reached $v_3$. Any such mutation apparently would
be temporary and thus not significant enough to induce speciation.
Indeed, if birds could give birth to dinosaurs, we would have to revise our basic
ideas about bird and dinosaur species.

The convexity axiom is similar to Dollo's law of irreversibility.
Dollo wrote in a letter: ``Nautilus does not yet have fins; Octopus has
them no longer. From this point of view the series ends where it begins.
But there is nothing contrary to irreversibility in this. After all,
Octopus has not turned into a nautiloid''
\cite{gould1970dollo}. The axiom is a compromise between two naive
errors: (i) that
every organism is the same species as its parent; (ii) that every organism is
the same species as its child. Either of these two stronger assertions has the
unintended consequence of implying that all organisms are the same species
(assuming all are descended from a single original organism). The convexity
axiom avoids this unintended consequence while salvaging much of the combinatorial
fruit of the naive versions.

\subsection{\red{Specieslike clusters}}

\begin{definition}
\label{specieslikeclusterdefn}
    \red{A set $S\subseteq G$ is a \emph{specieslike cluster} if
    the following conditions hold:}
    \begin{enumerate}
        \item \red{$S$ is connected.}
        \item \red{$S$ satisfies the identical ancestor point axiom.}
        \item \red{$S$ satisfies the convexity axiom.}
    \end{enumerate}
\end{definition}

\red{In the next section, we will show that specieslike clusters reduce the
subjectivity of species definition in a certain technical sense.
But first, we want to make a few remarks about small and large specieslike
clusters.}

\begin{remark}
\label{smallclustersrmk}
    \red{(Small specieslike clusters)
    The reader will notice there are certain degenerate examples of
    specieslike clusters. For any organism $v$, it is easy to check that
    the singleton $\{v\}$ is a specieslike cluster. This is usually
    undesirable, but it \emph{could} be desirable, if in fact $v$ is
    isolated from the rest of $G$ (picture: a novel organism arising
    from spontaneous biogenesis and then immediately dying).}
\end{remark}

\begin{remark}
\label{largeclustersrmk}
    \red{(Large specieslike clusters)
    In order to solve the problem of small specieslike clusters, the
    most natural thing to do would be to consider specieslike clusters
    which are maximal (that is, sets which are specieslike clusters,
    but which have no specieslike cluster proper supersets).
    Unfortunately, depending on $G$, there may be some vertices which
    do not inhabit any maximal specieslike cluster
    (\redder{for example, the leftmost vertex in
    Figure \ref{bigclusercounterexamplefig} does not belong to any
    maximal specieslike cluster\footnote{\redder{To see this, note
    that if such a specieslike cluster contained only finitely many vertices from
    the top branch in the figure, then it could be enlarged by
    adding one more vertex of that branch. Thus, a maximal such cluster
    would have to contain infinitely many vertices from the top branch, and so
    by convexity would contain the entire top branch. By identical reasoning,
    it would contain the entire bottom branch as well.
    But then it would clearly not satisfy the identical ancestor point axiom.}}}).
    This would make
    maximal specieslike clusters undesirable as a species notion if
    we assume that every organism should inhabit some species.}
\end{remark}

\begin{figure}
    \begin{center}
        \includegraphics[scale=0.5]{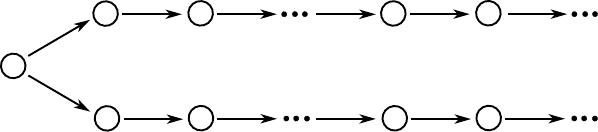}
    \end{center}
    \caption{\redder{An example of an infinitary biosphere in which a vertex fails
    to be contained in any maximal specieslike cluster.}}
    \label{bigclusercounterexamplefig}
\end{figure}

\red{Remark \ref{largeclustersrmk} naturally leads to the following question.}

\begin{informalquestion}
\label{mainqstn}
    \red{(To be answered in Section \ref{inclusivespeciessectn})
    Are there some additional constraints (plausible with species in
    mind) that can be imposed on
    specieslike clusters in order to ensure that every organism
    necessarily inhabits some maximal specieslike cluster satisfying
    said additional constraints? Of course this question is informal since
    ``plausible with species in mind'' is informal.
    In Section \ref{genericitysection} we will introduce two additional
    constraints, which we call the \emph{common ancestor property} and the
    \emph{reflection property}, and in Section \ref{inclusivespeciessectn}
    we will prove that every organism inhabits a maximal specieslike
    cluster satisfying these two additional properties.}
\end{informalquestion}

\section{Objective Species}
\label{objectivespeciessection}

The question of how to classify organisms into species seems inherently
subjective. One could contrive edge-cases where the same organism might get
classified into very different species depending on certain arbitrary
choices in the definition of species. In this section, we will show that
if we constrain species to \red{be specieslike clusters (as defined in
Definition \ref{specieslikeclusterdefn})}, then in a technical
sense, this subjectivity is reduced.

\begin{definition}
    Suppose $S\subseteq G$\redder{.}
    \red{An organism $v\in S$ is a \emph{generator} of $S$ if $S$
    contains at most finitely many non-descendants of $v$.}
    We will write $\mathrm{gen}(S)$ for the set of generators of $S$.
\end{definition}

\begin{lemma}
\label{generatorisequivalenttonotmule}
    If $S\subseteq G$ satisfies the identical ancestor point axiom,
    then for every $v\in S$,
    \red{if $v$ has infinitely many descendants in $S$,
    then $v$ is a generator of $S$.}
    % either $v$ is a generator of $S$,
    % or $v$ is a higher-order mule of $S$.
\end{lemma}

\begin{proof}
    Straightforward.
\end{proof}

\begin{definition}
\label{simequivalencerelndefn}
    Suppose $S_1,S_2\subseteq G$ are infinite.
    We say $S_1\sim S_2$ if the symmetric difference
    $\mathrm{gen}(S_1)\triangle \mathrm{gen}(S_2)$ is finite.
\end{definition}

\red{The statement that the generator-sets of
$S_1$ and $S_2$ have finite
symmetric difference can be informally glossed as
``$S_1$ and $S_2$ have almost the exact same generator-set''.
This makes it intuitively clear that $\sim$ is an equivalence
relation (we leave the formal details to the reader).}

\red{It can be shown that $S_1\sim S_2$ iff there exists some $r\in\mathbb R$
such that $\mathrm{gen}(S_1)$
and $\mathrm{gen}(S_2)$ contain exactly the same organisms born after
time $r$.}

\begin{theorem}
\label{objspeciesthm}
    (Objective Species Theorem)
    Suppose $S_1,S_2\subseteq G$ are \red{infinite specieslike clusters}.
    If $\mathrm{gen}(S_1)\cap\mathrm{gen}(S_2)$ is infinite
    then $S_1\sim S_2$.
\end{theorem}

\begin{proof}
    Assume the hypotheses. For sake of contradiction, assume
    $S_1\not\sim S_2$, i.e., that
    $\mathrm{gen}(S_1)\triangle \mathrm{gen}(S_2)$ is infinite.
    Without loss of generality, we may assume
    $\mathrm{gen}(S_1)\backslash \mathrm{gen}(S_2)$ is infinite,
    let its members be $v_1,v_2,\ldots$.
    Let $v\in \mathrm{gen}(S_1)\cap \mathrm{gen}(S_2)$.
    Since $v$ is a generator of $S_1$ and $\{v_1,v_2,\ldots\}$
    is an infinite set of members of $S_1$,
    there is some $i$ such that $v$ is an ancestor of
    $v_i$. Since $v_i$ is a generator of $S_1$ and
    $\mathrm{gen}(S_1)\cap \mathrm{gen}(S_2)$ is an infinite
    subset of $S_1$, $v_i$ has some descendant
    $w\in \mathrm{gen}(S_1)\cap\mathrm{gen}(S_2)$.
    In particular, $w\in S_2$.
    Since $v_i$ has an ancestor in $S_2$ (namely $v$)
    and a descendant in $S_2$ (namely $w$), and since
    $S_2$ satisfies the convexity axiom, we have $v_i\in S_2$.
    \redder{Since $v_i$ is an ancestor of $w$, and $w$ has infinitely many
        descendants in $S_2$ (since $w\in\mathrm{gen}(S_2)$), it follows that
        $v_i$ has infinitely many descendants in $S_2$.}
    % Since $v_i$ is a generator of $S_1$, it follows that $v_i$
    % is an ancestor of infinitely many members of
    % $\mathrm{gen}(S_1)\cap\mathrm{gen}(S_2)$,
    % all of whom are members of $S_2$.
    % Thus $v_i$ is not a higher-order mule of $S_2$.
    By Lemma \ref{generatorisequivalenttonotmule},
    $v_i\in \mathrm{gen}(S_2)$.
    Absurd: $v_1,v_2,\ldots$ are the members of
    $\mathrm{gen}(S_1)\backslash \mathrm{gen}(S_2)$.
\end{proof}

Thus, in a sense, the arbitrariness of species classification is
reduced as long as we require
species \red{be specieslike clusters (Definition \ref{specieslikeclusterdefn})}.
If we only care about species up to $\sim$-equivalence classes,
and if we consider \red{the identical ancestor point axiom\redder{,} the convexity
axiom}\redder{, and connectivity}
as reasonable assumptions for species,
then no two distinct species (modulo $\sim$)
can have significant \red{(i.e., infinite)} overlap, except possibly
\red{for non-generators.}
% (as in the case of
% horses and donkeys) for (possibly higher-order) mules.

\red{Non-generators play a role analogous to the joint offspring
of horses and donkeys, namely \emph{mules}. In some biological contexts,
it is convenient to assume hybrids do not
exist. In such contexts, the question inevitably arises: what
about mules? The standard answer is that mules can be ignored because they are
sterile. Lemma \ref{generatorisequivalenttonotmule} suggests that
non-generators are ``sterile'' in a weaker sense, and thus
``can be ignored'' in an analogous way. When we ignore them and focus
solely on generators, Theorem \ref{objspeciesthm} becomes a strong statement
about objective species: any two infinite specieslike clusters have
almost no overlap or almost total overlap.}

\section{Generic Properties of Organism-Sets}
\label{genericitysection}

\red{Toward the goal of answering Informal Question \ref{mainqstn},
w}e will develop a general framework for
guaranteeing maxima or minima of $\sim$-equivalence classes.
We will define notions of \emph{upward genericity} and \emph{downward genericity}
of sets of $G$-subsets.
Upward generic sets of $G$-subsets lend themselves to construction
of maximal $G$-subsets. Downward generic sets of $G$-subsets lend themselves to
construction of minimal $G$-subsets.
These notions resemble notions of genericity from \cite{alexander2020self}.

If $(X,<)$ is a linear order, then an \emph{ascending (resp.\ descending)
chain of $G$-subsets
indexed by $X$} is a family $\{S_\alpha\}_{\alpha\in X}$
(each $S_\alpha\subseteq G$)
such that for all $\alpha,\beta\in X$, if $\alpha<\beta$ then
$S_\alpha\subseteq S_\beta$ (resp.\ $S_\beta\subseteq S_\alpha$).
An \emph{ascending (resp.\ descending) chain} of $G$-subsets is an
ascending (resp.\ descending) chain of $G$-subsets indexed by some nonempty
linear order $(X,<)$.

\begin{definition}
    (Genericity)
    Suppose $T$ is a set of $G$-subsets.
    \begin{itemize}
        \item
        $T$ is \emph{upward generic} if the following property holds.
        For every ascending chain $\{C_\alpha\}_{\alpha\in X}$
        of nonempty $G$-subsets,
        if each $C_\alpha\in T$, then $\cup_\alpha C_\alpha\in T$.
        \item
        $T$ is \emph{downward generic} if the following property holds.
        For every descending chain $\{C_\alpha\}_{\alpha\in X}$
        of nonempty $G$-subsets,
        if each $C_\alpha\in T$, then $\cap_\alpha C_\alpha\in T$.
    \end{itemize}
\end{definition}

\begin{lemma}
\label{closureunderwedgelemma}
    (Closure under $\bigcap$)
    Suppose $\mathcal T$ is a set of sets of $G$-subsets.
    \begin{enumerate}
        \item
        If every $A\in\mathcal T$ is upward generic,
        then $\bigcap\mathcal T$ is upward generic.
        \item
        If every $A\in\mathcal T$ is downward generic,
        then $\bigcap\mathcal T$ is downward generic.
    \end{enumerate}
\end{lemma}

\begin{proof}
    (1)
    Let $\{C_\alpha\}_{\alpha\in X}$ be an ascending chain
    of nonempty
    $G$-subsets such that each $C_\alpha\in\bigcap\mathcal T$.
    Let $C=\cup_\alpha \mbox{\red{$C$}}_\alpha$, we must show
    $C\in\bigcap\mathcal T$.
    In other words, we must show that for every $A\in\mathcal T$,
    $C\in A$.
    Let $A\in\mathcal T$.
    Since each $C_\alpha\in\bigcap\mathcal T$,
    it follows that each $C_\alpha\in A$.
    Thus, since $A$ is upward generic, $C\in A$, as desired.

    (2) Similar to (1).
\end{proof}

\begin{theorem}
\label{abstractspeciesexistencetheorem}
    (Existence of Extrema)
    Suppose $T$ is a nonempty set of $G$-subsets.
    \begin{enumerate}
        \item
        If $T$ is upward generic then $T$ contains a maximal element
        (a $G$-subset which is not a proper subset of any element of $T$).
        Furthermore, for any $S\in T$,
        such a maximal element can be found which is a superset of $S$.
        \item
        If $T$ is downward generic
        then $T$ contains a minimal element
        (a $G$-subset which is not a proper superset of any member of $T$).
        Furthermore, for any $S\in T$, such a minimal element can be
        found which is a subset of $S$.
    \end{enumerate}
\end{theorem}

\begin{proof}
    By Zorn's lemma.
\end{proof}

\subsection{Positive genericity results}

\begin{definition}
\label{CAdefn}
    A set $S\subseteq G$ of organisms \red{has} the
    \emph{common ancestor \red{property}} if there exists some $v\in S$ (called
    the \emph{common ancestor} of $S$) such that for all $w\in S$,
    if $w\not=v$ then $v$ is an ancestor of $w$.
\end{definition}

\red{The common ancestor property has a cladistical
flavor. Combined with convexity, it says that $S$ is monophyletic or
paraphyletic. It is arguably plausible as a species property, at least if
we enforce the convention that species must have definite starting points
(the same way we declare a certain day to be the first day of Summer,
even though not particularly distinct from Spring).
It follows from Definition \ref{infinitebiospheredefn} (part 1) that
if a set has the common ancestor property, then its common ancestor
is unique.}

\begin{definition}
\label{reflectiondefn}
    A set $S\subseteq G$ of organisms \red{has} the \emph{reflection \red{property}}
    if the following statement is true. For every organism $v\in S$, if
    $v$ has infinitely many descendants, then $v$ has infinitely many
    descendants in $S$.
\end{definition}

The reflection \red{property}
is motivated by reflection properties from mathematical
logic \cite{carlson2001elementary}.
\red{Species would have the reflection property provided that in delineating
them, we enforced the following convention: whenever}
an organism
has many descendants in species $A$ and few (or no) descendants in species
$B$, then $A$ \red{is the} better choice for classifying
that organism itself, rather than $B$.

\begin{definition}
    We assign names to the following sets of $G$-subsets.
    \begin{itemize}
        \item $\mathrm{IAP}$ is the set of $G$-subsets which
            satisfy the identical ancestor point axiom
            (Definition \ref{udadefn}).
        \item $\mathrm{CONV}$ is the set of $G$-subsets which
            satisfy the convexity axiom
            (Definition \ref{convexitydefn}).
        \item $\mathrm{CA}$ is the set of $G$-subsets which
            \red{have} the common ancestor \red{property}
            (Definition \ref{CAdefn}).
        \item $\mathrm{REF}$ is the set of $G$-subsets which
            \red{have} the reflection \red{property}
            (Definition \ref{reflectiondefn}).
    \end{itemize}
\end{definition}

\begin{theorem}
\label{upwardgenericpositiveresultslemma}
    (Some upward genericity results)
    \begin{enumerate}
        \item
        $\mathrm{CONV}$ is upward generic.
        \item
        $\mathrm{CA}$ is upward generic.
        \item
        $\mathrm{REF}$ is upward generic.
        \item
        $\mathrm{IAP}\cap\mathrm{CONV}\cap\mathrm{CA}\cap\mathrm{REF}$
        is upward generic.
    \end{enumerate}
\end{theorem}

\begin{proof}
    For each case below, let $A$ be the set of $G$-subsets in question
    (for example, in the first case, $A=\mathrm{CONV}$).
    In each case, let $\mathscr C=\{C_\alpha\}_{\alpha\in X}$ be a nonempty
    ascending chain where each $C_\alpha\in A$,
    and let $C=\cup_\alpha C_\alpha$.

    \item
    ($\mathrm{CONV}$) We must show $C\in \mathrm{CONV}$.
    Suppose $v\in G$ has an ancestor $v'\in C$ and a descendant $v''\in C$.
    Since $\mathscr C$ is ascending, there is some $\alpha\in X$ such that
    $v',v''\in C_\alpha$. Since $C_\alpha\in \mathrm{CONV}$,
    $C_\alpha$ satisfies the convexity axiom.
    Thus $v\in C_\alpha$, thus $v\in C$.
    By arbitrariness of $v$, this shows $C$ satisfies the convexity axiom,
    i.e., $C\in\mathrm{CONV}$.

    \item
    ($\mathrm{CA}$) We must show $C\in\mathrm{CA}$.
    Assume not. Let $\alpha_0\in X$ be arbitrary.
    Since $C_{\alpha_0}\in\mathrm{CA}$,
    some $v_0\in C_{\alpha_0}$ is a common ancestor for $C_{\alpha_0}$.
    By assumption, $v_0$ is not a common ancestor for $C$, so there is
    some $\alpha_1\in X$ such that $v_0$ is not a common ancestor for
    $C_{\alpha_1}$. Clearly this implies $\alpha_1>\alpha_0$.
    Since $C_{\alpha_1}\in\mathrm{CA}$,
    some $v_1\in C_{\alpha_1}$ is a common ancestor for $C_{\alpha_1}$.
    Since $v_1\not=v_0$ and $C_{\alpha_0}\subseteq C_{\alpha_1}$,
    $v_1$ is an ancestor of $v_0$. This process can be continued
    forever, yielding a sequence $v_0,v_1,\ldots$ where each $v_{i+1}$
    is an ancestor of $v_i$. This contradicts the fact that only
    finitely many organisms were born before $v_0$
    (Definition \ref{infinitebiospheredefn}).

    \item
    ($\mathrm{REF}$) We must show $C\in\mathrm{REF}$.
    Suppose $v\in C$ has infinitely many descendants.
    Since $v\in C$, $v\in C_\alpha$ for some $\alpha$.
    Since $C_\alpha\in\mathrm{REF}$,
    $v$ has infinitely many descendants in $C_\alpha$.
    Since $C\supseteq C_\alpha$, we are done.

    \item
    ($\mathrm{IAP}\cap\mathrm{CONV}\cap\mathrm{CA}\cap\mathrm{REF}$)
    By (1--3), $C\in\mathrm{CONV}\cap\mathrm{CA}\cap\mathrm{REF}$.
    It remains to show $C\in\mathrm{IAP}$.
    Assume not. Then there is some $v\in C$ such that $C$ contains
    infinitely many descendants of $v$ and infinitely many non-descendants
    of $v$.
    Let $w_1$ be the common ancestor of $C$.

    Claim: There is an infinite path $Q=w_1,w_2,\ldots$ in $C$ of
    non-descendants of $v$ starting at $w_1$. The proof is like that
    of K\"onig's lemma (Lemma \ref{koniglemma}).
    Since $w_1$ has only finitely many children
    (Definition \ref{infinitebiospheredefn}) and $w_1$ has infinitely
    many descendants in $C$ which are non-descendants of $v$, some child
    $w_2$ of $w_1$ must have infinitely many descendants in $C$ which
    are non-descendants of $v$. Since $C\in\mathrm{CONV}$,
    $w_2\in C$; $w_2$ is a non-descendant of $v$ lest all of $w_2$'s
    descendants would be descendants of $v$. Similarly, since $w_2$ has
    only finitely many children and $w_2$ has infinitely many descendants
    in $C$ which are non-descendants of $v$, some child $w_3$ of $w_2$
    must have infinitely many descendants in $C$ which are non-descendants
    of $v$. This process continues forever, producing the desired
    path $w_1,w_2,\ldots$, proving the claim.

    \red{Since $C$ is the increasing union of the $C_\alpha$,
    there is some $\alpha\in X$ such that $v,w_1\in C_\alpha$.}
    % Now let $\alpha\in X$ be such that $v,w_1\in C_\alpha$.
    We claim that $w_n\in C_\alpha$ for all $n$\redder{. Fix any $n>1$.}
    Let $\beta\in X$ be such that $\redder{w_{n}}\in C_\beta$, we may assume
    $\beta\geq\alpha$.
    \redder{Since $w_1$ has infinitely many descendants (namely
    $w_2,w_3,\ldots$), the reflection property for $C_\alpha$ guarantees
    $w_1$ has infinitely many descendants in $C_\alpha$, thus
    $|C_\alpha|=\infty$. Since $w_n$ has infinitely many descendants
    (namely $w_{n+1},w_{n+2},\ldots$), the reflection property for $C_\beta$
    guarantees $w_n$ has infinitely many descendants in $C_\beta$.
    By the identical ancestor point axiom for $C_\beta$, it follows that
    $C_\beta$ cannot contain infinitely many non-descendants of $w_n$.
    But $C_\alpha\subseteq C_\beta$ since $\mathscr C$ is an ascending chain.
    Thus since $|C_\alpha|=\infty$, $C_\alpha$ must contain some descendant
    of $w_n$. Since $w_n$ has a descendant in $C_\alpha$, and also an ancestor
    (namely $w_1$) in $C_\alpha$, the convexity axiom for $C_\alpha$
    forces $w_n\in C_\alpha$, as claimed.}
    % ; we will prove this
    % by induction on $n$ (see Figure \ref{inductionfigure}).
    % The base case is given. Assume $w_n\in C_\alpha$.
    % Let $\beta\in X$ be such that $w_{n+1}\in C_\beta$, we may assume
    % $\beta\geq\alpha$.
    % Since $w_n$ has infinitely many descendants (namely $w_{n+1},w_{n+2},\ldots$),
    % the reflection \red{property} for $C_\alpha$ ensures $w_n$ has infinitely
    % many descendants in $C_\alpha$.
    % By a similar argument, $w_{n+1}$ has infinitely many descendants in
    % $C_\beta$.
    % \redder{Since $C_\beta$ satisfies the identical ancestor point axiom,
    % $C_\beta$ contains at most finitely many non-descendants of $w_{n+1}$.}
    % % By the identical ancestor point axiom for $C_\beta$,
    % % $w_{n+1}$ can have at most
    % % finitely many non-descendants in $C_\beta$.
    % Thus,
    % since $C_\alpha\subseteq C_\beta$, one of $w_n$'s infinitely many
    % descendants $d$ in $C_\alpha$ must be a descendant of $w_{n+1}$.
    % So $w_{n+1}$ has a descendant in $C_\alpha$ (namely $d$) and an ancestor
    % in $C_\alpha$ (namely $w_n$). By the convexity axiom for $C_\alpha$,
    % $w_{n+1}\in C_\alpha$. This completes the inductive argument.

    Since $v$ has infinitely many descendants and $v\in C_\alpha$,
    the reflection \red{property} for $C_\alpha$ ensures $v$ has infinitely
    many descendants in $C_\alpha$.
    But $v$ also has infinitely many non-descendants in $C_\alpha$,
    namely $w_1,w_2,\ldots$.
    This contradicts the identical ancestor point axiom for $C_\alpha$.
\end{proof}

% \begin{figure}
%     \begin{center}
%         \includegraphics[scale=0.5]{induction.pdf}
%     \end{center}
%     \caption{Proof of a claim in the proof of the
%     $\mathrm{IAP}\cap\mathrm{CONV}\cap\mathrm{CA}\cap\mathrm{REF}$
%     part of Theorem \ref{upwardgenericpositiveresultslemma}.}
%     \label{inductionfigure}
% \end{figure}

\begin{theorem}
\label{downwardgenericityresultslemma}
    (Some downward genericity results)
    \begin{enumerate}
        \item
        $\mathrm{CONV}$ is downward generic.
        \item
        $\mathrm{CONV}\cap\mathrm{REF}$ is downward generic.
        \item
        $\mathrm{IAP}$ is downward generic.
    \end{enumerate}
\end{theorem}

\begin{proof}
    For each case below, let $A$ be the set of $G$-subsets in question (for example,
    in the first case, $A=\mathrm{CONV}$).
    In each case, let $\mathscr C=\{C_\alpha\}_{\alpha\in X}$ be a nonempty
    descending chain where each $C_\alpha\in A$,
    and let $C=\cap_\alpha C_\alpha$.

    \item
    ($\mathrm{CONV}$)
    Suppose organism $v$ has an ancestor $v'\in C$ and a descendant
    $v''\in C$.
    For every $\alpha\in X$, $v'\in C_\alpha$ and $v''\in C_\alpha$,
    thus $v\in C_\alpha$ by the convexity axiom for $C_\alpha$.
    Thus $v\in C$.

    \item
    ($\mathrm{CONV}\cap\mathrm{REF}$)
    By (1), $C$ satisfies the convexity axiom.
    To show $C$ \red{has} the reflection \red{property}, suppose $v\in C$ has
    infinitely many descendants.
    \red{Since $v\in C$, $v$ is in $C_\alpha$ for every
    $\alpha\in X$.}
    For each \red{$\alpha\in X$,}
    the reflection \red{property} for $C_\alpha$
    ensures $v$ has infinitely many descendants in $C_\alpha$.
    \red{Now we will embark on a somewhat involved argument in the style of
    K\"onig's Lemma.}

    \red{\textbf{Claim:} $v$ has a child $v_1$ such that for each $\alpha\in X$,
    $v_1$ has infinitely many descendants in $C_\alpha$.}

    \redder{Assume the claim fails. Then for each child $w$ of $v$,
    there is some $\alpha_w$ such that $w$ has at most finitely many
    descendants in $C_{\alpha_w}$. Since $\mathscr C$ is a descending chain,
    this implies $w$ has at most finitely many descendants in $C_\beta$
    for every $\beta\geq\alpha_w$. Let $\beta=\max_{w\in W}\alpha_w$
    where $W$ is the set of all children of $v$ (the maximum is defined
    since $W$ is finite). By construction,
    each child of $v$ has at most finitely many descendants in $C_\beta$.
    Thus $v$ itself has at most finitely many descendants in $C_\beta$,
    contradicting that $v$ has infinitely many descendants in $C_\alpha$ for
    every $\alpha$. By contradiction, the claim is proved.}

    % \red{We prove the claim as follows. For each $\alpha\in X$,
    % since $v$ has infinitely many descendants in $C_\alpha$,
    % and $v$ has only finitely many children,
    % some child $v_{1\alpha}$ of $v$ must have infinitely many descendants
    % in $C_\alpha$. Now, since $\mathscr C$ is a descending chain,
    % for each $\beta\in X$, we have $C_\beta\subseteq C_\alpha$ for all
    % $\alpha\leq\beta$; thus $v_{1\beta}$ has infinitely many descendants in
    % $C_\alpha$ for all $\alpha\leq\beta$. If $X$ has a maximum element
    % $\gamma$, it follows that $v_1=v_{1\gamma}$ witnesses the claim.
    % But assume $X$ has no maximum. Since $v$ has only finitely many children,
    % some child $v_1^{cof}$ of $v$ must equal $v_{1\beta}$ for cofinally many
    % $\beta$ (by which we mean that for all $\alpha\in X$, there is
    % some $\beta>\alpha$ in $X$ such that $v_1^{cof}=v_{1\beta}$)---if none
    % of the children $w_1,\ldots,w_k$ of $v$ equals $v_{1\beta}$ for cofinally
    % many $\beta$ then for each $i=1,\ldots,k$ there must be some
    % $\alpha_i\in X$ such that there is no $\beta>\alpha_i$ with
    % $v_{1\beta}=w_i$, but then since $X$ has no maximum,
    % there would be some $\gamma\in X$ with $\gamma>\alpha_i$ for
    % all $i=1,\ldots,k$, and then $v_{1\gamma}$ could not be any child
    % $w_1,\ldots,w_k$ of $v$, a contradiction.
    % So $v_1=v_1^{cof}$ witnesses the claim.}

    \red{By identical reasoning,
    $v_1$ has a child $v_2$ such that for each $\alpha\in X$,
    $v_2$ has infinitely many descendants in $C_\alpha$.
    And so on forever, defining a path $v,v_1,v_2,\ldots$
    such that for each $i$, for each $\alpha\in X$,
    $v_i$ has infinitely many descendants in $C_\alpha$.}
    For each $\alpha$, the convexity axiom for $C_\alpha$ forces
    each $v_i\in C_\alpha$. So each $v_i\in C$, so $v$ has infinitely
    many descendants in $C$.

    \item
    ($\mathrm{IAP}$)
    Suppose $C$ fails to satisfy the identical ancestor point axiom.
    Then there is some
    $v\in C$ such that $C$ contains infinitely many descendants of
    $v$ and infinitely many non-descendants of $v$.
    Fix some $\alpha\in X$. Since $C_\alpha\supseteq C$,
    $v\in C_\alpha$ and
    $C_\alpha$ contains infinitely many descendants of $v$ and
    infinitely many non-descendants of $v$.
    This violates the identical ancestor point axiom for $C_\alpha$.
\end{proof}

\subsection{Negative genericity results}

In this subsection we will show that in one sense
Theorems \ref{upwardgenericpositiveresultslemma}
and \ref{downwardgenericityresultslemma} are as strong as possible.

\begin{proposition}
\label{firstnegativeprop}
    (Compare Theorem \ref{upwardgenericpositiveresultslemma})
    If $X$ is a proper subset of
    \[\{\mathrm{CONV},\mathrm{CA},\mathrm{REF},\mathrm{IAP}\},\]
    and $X$ contains $\mathrm{IAP}$, then there is some infinite
    biosphere $G$ such that $\bigcap X$ is not
    upward generic.
\end{proposition}

\begin{proof}
    Case 1: $X$ fails to contain $\mathrm{REF}$.
    Let $G$ consist (see Figure \ref{negativefigure} part a)
    of two infinite directed paths $v_1,v_2,\ldots$ and
    $w_1,w_2,\ldots$ such that $v_1=w_1$ but such that
    $G\backslash\{v_1\}$ is disconnected.
    For each $\alpha\in\mathbb N$, let
    $S_\alpha=\{v_1,v_2,\ldots\}\cup\{w_1,\ldots,w_\alpha\}$.
    It is easy to show each $S_\alpha\in
    \mathrm{CONV}\cap\mathrm{CA}\cap\mathrm{IAP}$,
    but $S=\cup_\alpha S_\alpha\not\in\mathrm{IAP}$.

    Case 2: $X$ fails to contain $\mathrm{CA}$.
    Let $G$ consist (see Figure \ref{negativefigure} part b)
    of an infinite directed path $v_1,v_2,\ldots$,
    along with additional distinct vertices $w_1,w_2,\ldots$
    such that each $w_i$ is parentless and has $v_i$ as lone child.
    For each $\alpha\in \mathbb N$, let
    $S_\alpha=\{v_1,v_2,\ldots\}\cup\{w_1,\ldots,w_\alpha\}$.
    It is easy to show each
    $S_\alpha\in\mathrm{CONV}\cap\mathrm{REF}\cap\mathrm{IAP}$,
    but $S=\cup_\alpha S_\alpha\not\in\mathrm{IAP}$.

    Case 3: $X$ fails to contain $\mathrm{CONV}$.
    Let $G$ consist (see Figure \ref{negativefigure} part c)
    of two infinite directed paths $v_1,v_2,\ldots$ and
    $w_1,w_2,\ldots$ such that $v_1=w_1$, along with distinct additional
    vertices $u_2,u_3,\ldots$, such that:
    \begin{itemize}
        \item
        Each $u_i$ has lone parent $w_i$.
        \item
        Each $u_i$ has lone child $v_{i+1}$.
        \item
        $G\backslash (\{v_1\}\cup\{u_2,u_3,\ldots\})$ is disconnected.
    \end{itemize}
    For each $\alpha\in\mathbb N^+$, let
    \[
        S_\alpha
        =
        \{v_1,v_2,\ldots\}
        \cup
        \{w_1,\ldots,w_\alpha\}.
    \]
    It is easy to show that each
    $S_\alpha\in\mathrm{CA}\cap\mathrm{REF}\cap\mathrm{IAP}$,
    but $S=\cup_\alpha S_\alpha\not\in\mathrm{IAP}$.
\end{proof}

\begin{figure}
    \begin{center}
        \includegraphics[scale=1]{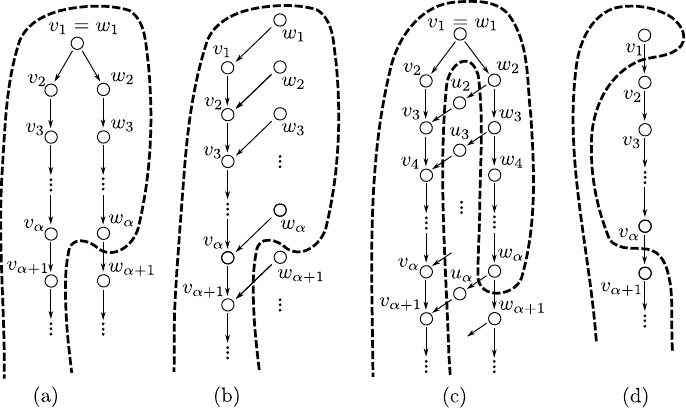}
    \end{center}
    \caption{Negative genericity examples. (a) The example in Case 1 of
    Proposition \ref{firstnegativeprop}. (b) The example in Case 2 of
    Proposition \ref{firstnegativeprop}. (c) The example in Case 3 of
    Proposition \ref{firstnegativeprop}. (d) The example in
    Proposition \ref{secondnegativeprop}.}
    \label{negativefigure}
\end{figure}

\begin{proposition}
\label{secondnegativeprop}
    (Compare Theorem \ref{downwardgenericityresultslemma})
    There exists an infinite biosphere $G$ such that
    $\mathrm{REF}$ is not downward generic.
\end{proposition}

\begin{proof}
    Let $G$ consist (see Figure \ref{negativefigure} part d)
    of an infinite directed path $v_1,v_2,\ldots$.
    For each $\alpha\in\mathbb N$, let
    $S_\alpha=\{v_1\}\cup\{v_i\,:\,i>\alpha\}$.
    It is easy to see that each $S_\alpha\in\mathrm{REF}$
    but $S=\cap_\alpha S_\alpha=\{v_1\}\not\in\mathrm{REF}$.
\end{proof}

\section{An Answer to Informal Question \ref{mainqstn}}
\label{inclusivespeciessectn}

% One important desideratum of species is that every organism should belong to
% a species. From Section \ref{objectivespeciessection} we know that if we
% require species to satisfy the identical ancestor point axiom and the
% convexity axiom
% then this determines species,
% in the sense that if any two species have infinite overlap
% of \red{generators}, then those two species
% are the same modulo $\sim$.
% It remains to choose particular $\sim$-equivalence classes, and we
% will do so by choosing them as large as possible, but in order to prove such
% maximal equivalence classes exist, we need to further require
% the common ancestor axiom and the reflection axiom.
% We will show that the species thus obtained satisfy the desideratum that
% every organism inhabits some such species.

\begin{definition}
\label{atomicspeciesdefn}
    By an
    \red{$\mathrm{IAP}\cap\mathrm{CONV}\cap\mathrm{CA}\cap\mathrm{REF}$-maximal
    set}
    we mean a nonempty subset $S$ of $G$
    such that:
    \begin{enumerate}
        \item
        $S$ satisfies the identical ancestor point axiom.
        \item
        $S$ satisfies the convexity axiom.
        \item
        $S$ \red{has} the common ancestor \red{property}.
        \item
        $S$ \red{has} the reflection \red{property}.
        \item
        $S$ is maximal: no strict superset of $S$ satisfies 1--4.
    \end{enumerate}
\end{definition}

\red{Note that the common ancestor property implies connectivity,
thus every
$\mathrm{IAP}\cap\mathrm{CONV}\cap\mathrm{CA}\cap\mathrm{REF}$-maximal
set is a specieslike cluster.}

\red{One could in fact consider
    $\mathrm{IAP}\cap\mathrm{CONV}\cap\mathrm{CA}\cap\mathrm{REF}$-maximal
    sets to be species, thereby obtaining an abstract, formal species
    notion. \redder{B}y the following theorem, this would
    be a major improvement
over our earlier attempt, in \cite{alexander2013infinite}, to apply infinitary combinatorics to species definition\footnote{\red{There, we defined an \emph{infinitary species} to be
a set of organisms which is (i) infinite, (ii) ancestrally closed (i.e., it contains
every ancestor of every member of itself), and (iii) minimal with respect to
(i) and (ii).}}.
If we did this, we \redder{might} suffer the following unintended consequence:
all of our nonhuman ancestors would be included in the human species
\redder{(see Example \ref{atomicspeciesexample} part 2 below)}.
Mathematically, this flaw is negligible, since the set
of our nonhuman ancestors is finite. Nevertheless, biologists might
still consider this flaw fatal. Thus, it might be more appropriate
to consider
$\mathrm{IAP}\cap\mathrm{CONV}\cap\mathrm{CA}\cap\mathrm{REF}$-maximal
sets to be chains of species
\redder{(not to be confused with the chains of Section \ref{genericitysection})}.
For now, we choose to avoid such controversies.}

\begin{theorem}
\label{atomicspeciesexistencethm}
    \red{(A positive answer to Informal Question \ref{mainqstn})}
    For every $v\in G$, there is an
    \red{$\mathrm{IAP}\cap\mathrm{CONV}\cap\mathrm{CA}\cap\mathrm{REF}$-maximal
    set}
    containing $v$.
\end{theorem}

\begin{proof}
    Let $\mathcal T=\mathrm{IAP}\cap\mathrm{CONV}
        \cap\mathrm{CA}\cap\mathrm{REF}$.
    By Theorem \ref{upwardgenericpositiveresultslemma}, $\mathcal T$
    is upward generic.
    If we can exhibit some $S\in \mathcal T$
    such that $v\in S$,
    then the theorem is immediate by
    Theorem \ref{abstractspeciesexistencetheorem} (part 1).
    Exhibiting such an $S$ is surprisingly nontrivial in general.

    Let $v^*$ be the set of those $G$-subsets $S$ such that
    $v$ is the common ancestor of $S$ (in other words,
    $S\in v^*$ if and only if $v\in S$ and for every $w\in S$
    with $w\not=v$, $w$ is a descendant of $v$).
    Let $\mathcal T_0=\mathrm{CONV}\cap\mathrm{REF}\cap v^*$.

    It is straightforward to show that $v^*$ is downward generic.
    By Theorem \ref{downwardgenericityresultslemma}
    and Lemma \ref{closureunderwedgelemma}, $\mathcal T_0$ is downward
    generic.

    Claim 1: $\mathcal T_0\not=\emptyset$.
    Let $S_v$ be the set consisting of $v$ and all $v$'s
    descendants. It is straightforward to see $S_v\in\mathcal T_0$,
    proving the claim.

    By Theorem \ref{abstractspeciesexistencetheorem} (part 2),
    $\mathcal T_0$ contains a minimal element $S$.
    Since $v\in S$ (because $S\in v^*$), to complete the proof of the theorem,
    it will suffice to show $S\in\mathcal T$.
    Since $S\in v^*$, clearly $S\in\mathrm{CA}$.
    To show $S\in \mathcal T$, it remains only to show
    $S\in\mathrm{IAP}$.

    Assume $S\not\in\mathrm{IAP}$.

    Then there is
    some organism $v'\in S$ such that $S$ contains infinitely
    many descendants of $v'$ and infinitely many non-descendants of $v'$.

    Claim 2: There exists an infinite path $Q$ in $S$ starting at $v$,
    all of whose members are non-descendants of $v'$.
    This is another K\"onig's lemma-style argument.
    Since $v$ is the common ancestor of $S$ and $v$ has only
    finitely many children (Definition \ref{infinitebiospheredefn}),
    some child $v_1$ of $v$ must have infinitely many descendants in $S$
    that are non-descendants of $v'$. So $v_1$ must be a non-descendant
    of $v'$ lest all of $v_1$'s descendants would be descendants of $v'$.
    Likewise, since $v_1$ has only finitely many children, $v_1$ must have
    some child $v_2$ which has infinitely many descendants in $S$ that
    are non-descendants of $v'$. And so on forever, defining a path
    $v,v_1,v_2,\ldots$. Since $S$ satisfies the convexity axiom,
    each $v_i\in S$. This proves the claim.

    \red{Let $T$ consist of $Q$ along with all organisms that have both
    an ancestor and a descendant in $Q$.}

    Claim 3: $T$ contains no descendant of $v'$.
    
    \red{By construction, $Q$ contains no descendant of $v'$, so any
    descendant $w$ of $v'$ in $T$ would have to have both an ancestor
    $w'$ and a descendant $w''$ in $Q$.
    But then $w''\in Q$ would be a descendant of $v'$, absurd.}

    Claim 4: $T\in\mathrm{CONV}\cap\mathrm{REF}\cap v^*$.

    The convexity axiom and the inclusion of $v$ as common ancestor
    are straightforward.
    For the reflection \red{property}, it is easy \red{to prove} that
    every vertex in $T$ has infinitely many descendants in \red{$Q$}.

    Claim 5: $T\subseteq S$. This follows by the definition of $T$ and
    the convexity axiom for $S$.

    Claims 3--5, and the fact that $S$ contains infinitely many
    descendants of $v'$, together contradict the minimality of $S$. Absurd.
\end{proof}

\red{We do not currently know of any qualitatively different sets of constraints
which would positively answer Informal Question \ref{mainqstn}. We would be
particularly interested in a solution not requiring the $\mathrm{CA}$ property.}

\begin{example}
\label{atomicspeciesexample}
    \red{Here are a few examples of
    $\mathrm{IAP}\cap\mathrm{CONV}\cap\mathrm{CA}\cap\mathrm{REF}$-maximal
    sets.}
    \begin{enumerate}
        \item
        \red{(No splits; see Figure \ref{example1figure})
        Suppose that $G$ is partitioned
        into sets $G_1,G_2,\ldots$ (thought of as \emph{generations}),
        each $G_i$ consisting of organisms $v_{i,1},\ldots,v_{i,k}$ (where
        $k$ is fixed), with an edge directed from $v_{i,p}$ to $v_{j,q}$ iff
        $j=i+1$ and $|p-q|\leq1$
        \redder{(if $k>2$ then some vertices have $>2$ parents, which is
        biologically unrealistic, but we chose this example for its
        mathematical elegance)}. Then $G$ contains exactly $k$
        $\mathrm{IAP}\cap\mathrm{CONV}\cap\mathrm{CA}\cap\mathrm{REF}$-maximal
        sets,
        each consisting of $v_{1,\ell}$ and all descendants of $v_{1,\ell}$
        (for $\ell\in\{1,\ldots,k\}$). These $k$ sets are all
        mutually equivalent modulo the equivalence relation $\sim$ of
        Section \ref{objectivespeciessection}.}
        \item
        \red{(A single two-way permanent split;
        see Figure \ref{example2figure}, left)
        Suppose $G$ has a common ancestor, and, after a certain time,
        permanently splits in two (with a two-way permanent split),
        and each sub-branch has its own common ancestor. Assuming no
        further splitting occurs, then $G$ will contain one
        $\mathrm{IAP}\cap\mathrm{CONV}\cap\mathrm{CA}\cap\mathrm{REF}$-maximal
        set
        for each branch. Both of these sets will also contain
        all organisms born before the split.
        The fact that the pre-split organisms inhabit
        both species simultaneously does not violate the Objective Species
        Theorem, because there are only finitely many such organisms.}
        \item
        \red{(A single one-way permanent split;
        see Figure \ref{example2figure}, right)
        Let $G$ be as in the previous example, except that the split in
        question is one-way instead of two-way. Then the
        $\mathrm{IAP}\cap\mathrm{CONV}\cap\mathrm{CA}\cap\mathrm{REF}$-maximal
        sets will
        be as in the previous example, except for the species corresponding
        to the branch receiving incoming arrows from the opposite branch.
        This set does \emph{not} contain pre-split organisms, because
        if it did, the convexity axiom would force it to also include vertices
        from the opposite branch, which would violate the identical ancestor
        point axiom.}
        \item
        \red{(A more interesting example)
        Figure \ref{example4figure} shows an infinite biosphere $G$ with
        both one- and two-way
        permanent splits---including splits nested inside the
        branches of other splits---and
        corresponding
        $\mathrm{IAP}\cap\mathrm{CONV}\cap\mathrm{CA}\cap\mathrm{REF}$-maximal
        sets.}
        \item
        \red{(Another type of one-way permanent split)
        Figure \ref{example5figure} shows an infinite biosphere $G$ with
        an inner population and two outer populations. Organisms from the
        two outer populations periodically contribute children to the
        inner population, which children do not initially have any parents
        from the inner population, but proceed to reproduce with their
        inner-population contemporaries. The figure also shows seven of
        the infinitely many resulting
        $\mathrm{IAP}\cap\mathrm{CONV}\cap\mathrm{CA}\cap\mathrm{REF}$-maximal
        sets. Each child contributed to
        the inner population by the two outer populations is necessarily
        the common ancestor of its own such set (it cannot share
        such a set with either of its parents, lest that set
        would include infinitely many members from all three
        populations---by the convexity axiom and reflection property---which would
        violate the identical ancestor point axiom). Modulo $\sim$, there
        are only three equivalence classes of
        $\mathrm{IAP}\cap\mathrm{CONV}\cap\mathrm{CA}\cap\mathrm{REF}$-maximal
        sets.}
    \end{enumerate}
\end{example}

\begin{figure}
    \begin{center}
        \includegraphics[scale=0.5]{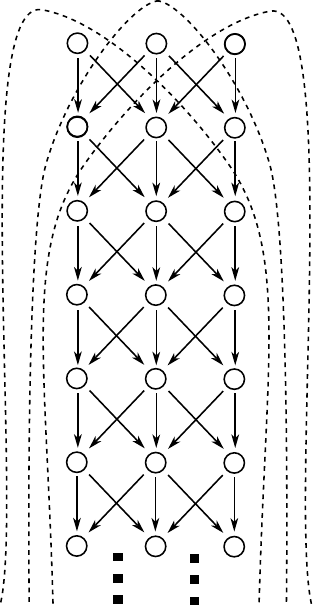}
    \end{center}
    \caption{\red{Three
    $\mathrm{IAP}\cap\mathrm{CONV}\cap\mathrm{CA}\cap\mathrm{REF}$-maximal
    sets, all equivalent (in the sense of
    Definition \ref{simequivalencerelndefn}), in the simple infinite biosphere of
    Example \ref{atomicspeciesexample} part 1 with $k=3$.
    \redder{The fact that some vertices have $>2$ parents is of course
    biologically unrealistic, but the example was chosen for its mathematical
    elegance.}}}
    \label{example1figure}
\end{figure}

\begin{figure}
    \begin{center}
        \includegraphics[scale=0.5]{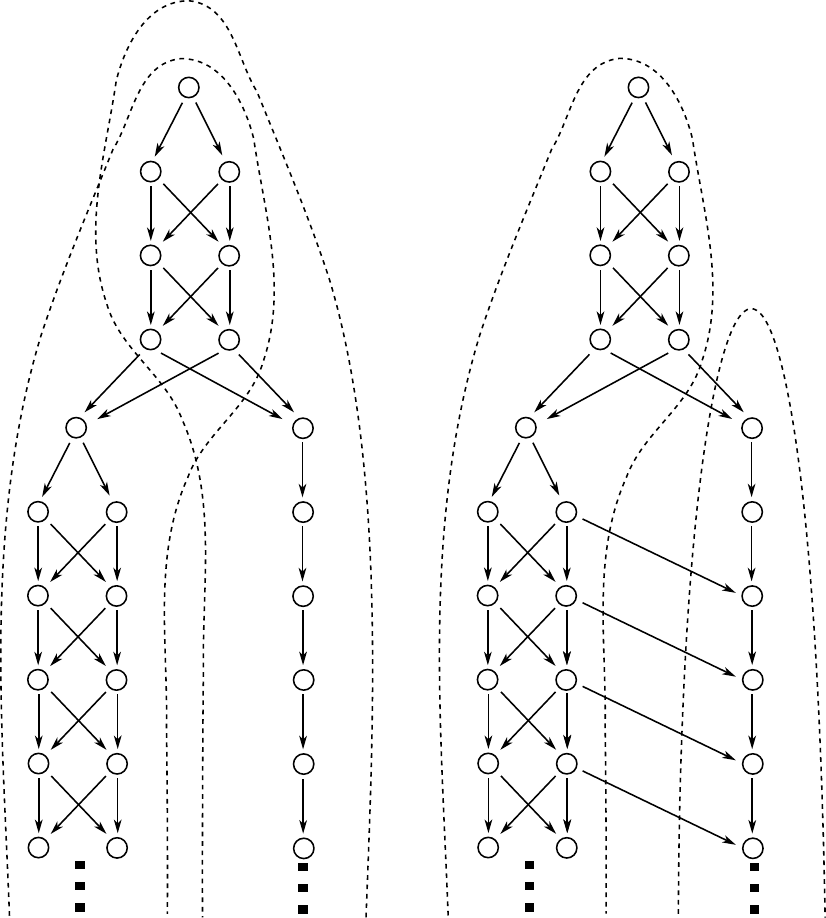}
    \end{center}
    \caption{\red{Left: an infinite biosphere consisting of a single two-way
    permanent split, and the two corresponding
    $\mathrm{IAP}\cap\mathrm{CONV}\cap\mathrm{CA}\cap\mathrm{REF}$-maximal sets.
    Right: an infinite biosphere consisting of a single one-way permanent
    split, and the two corresponding
    $\mathrm{IAP}\cap\mathrm{CONV}\cap\mathrm{CA}\cap\mathrm{REF}$-maximal sets.}}
    \label{example2figure}
\end{figure}

\begin{figure}
    \begin{center}
        \includegraphics[scale=0.5]{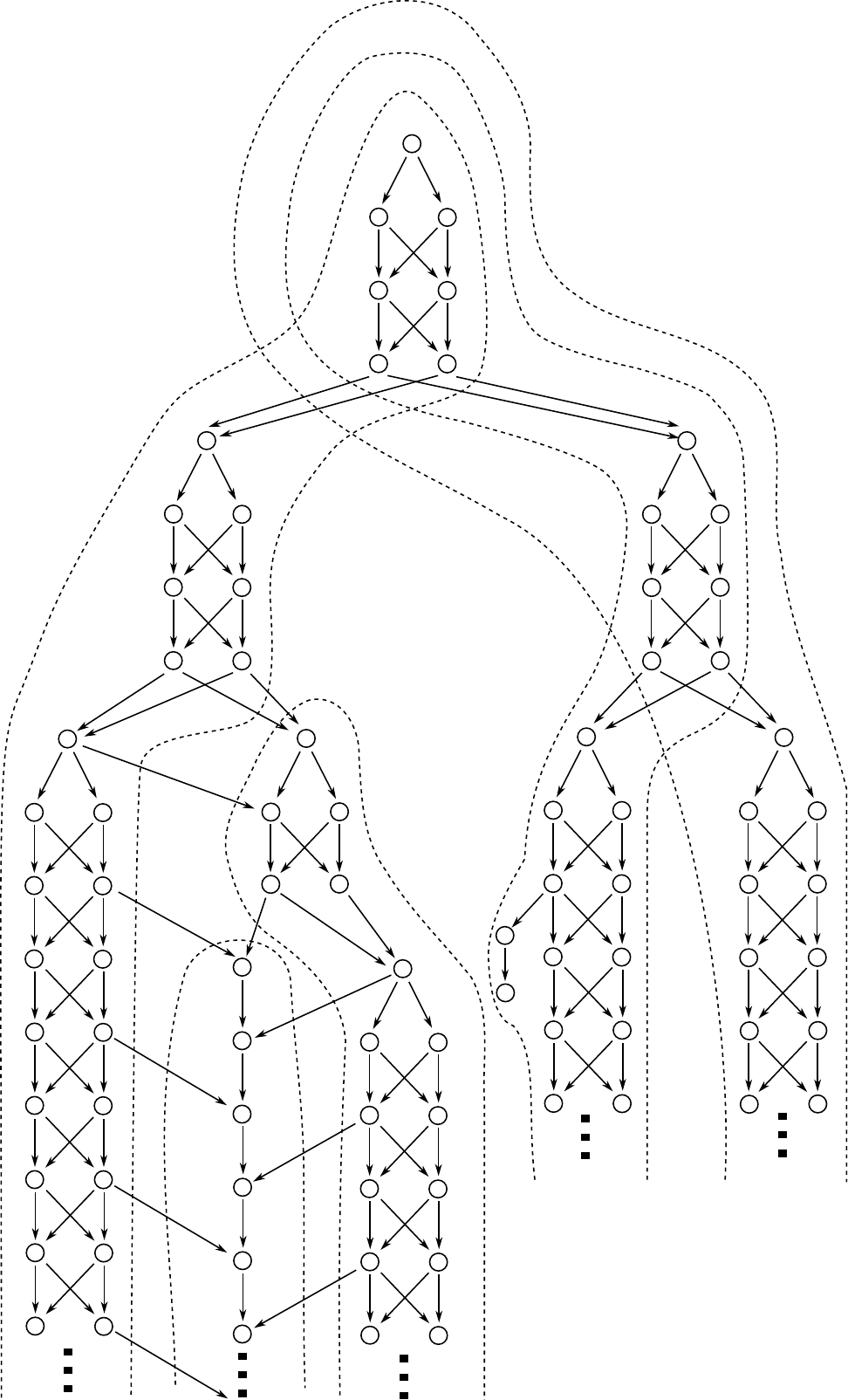}
    \end{center}
    \caption{\red{An infinite biosphere with various permanent splits
    (both one- and two-way), some of which are nested within others;
    and the five corresponding
    $\mathrm{IAP}\cap\mathrm{CONV}\cap\mathrm{CA}\cap\mathrm{REF}$-maximal
    sets.}}
    \label{example4figure}
\end{figure}

\begin{figure}
    \begin{center}
        \includegraphics[scale=0.5]{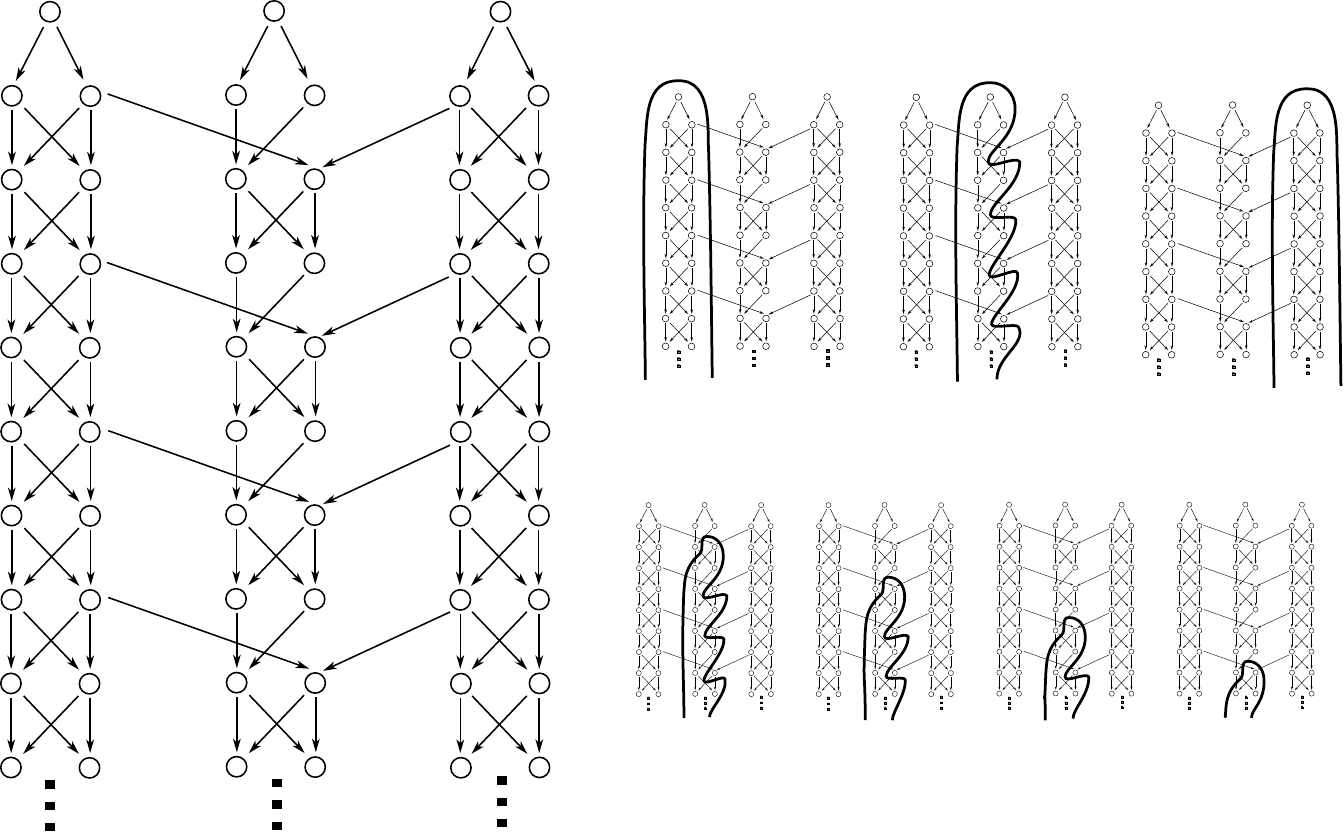}
    \end{center}
    \caption{\red{Left: An infinite biosphere made up of two outer
    populations and an inner population. Organisms from the outer
    populations periodically contribute children to the inner
    population, which children have no parents from the inner
    population. Right: Seven of the infinitely many resulting
    $\mathrm{IAP}\cap\mathrm{CONV}\cap\mathrm{CA}\cap\mathrm{REF}$-maximal
    sets (but modulo $\sim$, there are only three
    equivalence classes thereof).}}
    \label{example5figure}
\end{figure}

\red{It can be shown that
the set of
$\mathrm{IAP}\cap\mathrm{CONV}\cap\mathrm{CA}\cap\mathrm{REF}$-maximal
sets in $G$ (and the set of $\sim$-equivalence classes
thereof)
can be \redder{of} any positive integer \redder{size}
or can be countably or uncountably infinite.}

\red{Finally, we give examples to show that Theorem \ref{atomicspeciesexistencethm}
would fail if any one of $\mathrm{REF}$, $\mathrm{CA}$, or $\mathrm{CONV}$
were removed from Definition \ref{atomicspeciesdefn}.
The theorem would be trivially true but useless
if $\mathrm{IAP}$ were removed
(for then the maximal sets would be precisely the clades of
parentless organisms).}

\begin{example}
\label{necessityexample}
    \red{(Necessity of $\mathrm{CONV}$,
    $\mathrm{CA}$, and $\mathrm{REF}$)}
    \begin{itemize}
        \item
        \red{Call a set $S\subseteq G$ an
        \emph{$\mathrm{IAP}\cap\mathrm{CONV}\cap\mathrm{CA}$-maximal set}
        if $S$ is a maximal element of
        $\mathrm{IAP}\cap\mathrm{CONV}\cap\mathrm{CA}$.
        If $G$ is as in Figure \ref{negativefigure} part a,
        then $v_1$ is not a member of
        any $\mathrm{IAP}\cap\mathrm{CONV}\cap\mathrm{CA}$-maximal set.}
        \item
        \red{Call a set $S\subseteq G$ an
        \emph{$\mathrm{IAP}\cap\mathrm{CONV}\cap\mathrm{REF}$-maximal set}
        if $S$ is a maximal element of
        $\mathrm{IAP}\cap\mathrm{CONV}\cap\mathrm{REF}$.
        If $G$ is as in Figure \ref{negativefigure} part b,
        then $v_1$ is not a member of any
        $\mathrm{IAP}\cap\mathrm{CONV}\cap\mathrm{REF}$-maximal set.
        Alternatively, for a slightly more biologically plausible
        infinite biosphere (not involving infinitely many parentless
        organisms),
        the same statement holds for the top middle vertex in
        Figure \ref{example5figure}.}
        \item
        \red{Call a set $S\subseteq G$ an
        \emph{$\mathrm{IAP}\cap\mathrm{CA}\cap\mathrm{REF}$-maximal set}
        if $S$ is a maximal element of
        $\mathrm{IAP}\cap\mathrm{CA}\cap\mathrm{REF}$.
        If $G$ is as in Figure \ref{negativefigure} part c,
        then $v_2$ is not a member of any
        $\mathrm{IAP}\cap\mathrm{CA}\cap\mathrm{REF}$-maximal set.}
    \end{itemize}
\end{example}

\section{Replies to Anticipated Objections against the Identical Ancestor
Point Axiom}
\label{objectionssectn}

\subsection{Intuitive arguments}

It seems at least plausible that the reader's great-grandfather's
descendants might perpetuate indefinitely. But it does not seem plausible
that the reader's great-grandfather's descendants will eventually include
all living humans, as there is nothing particularly special about the
reader's great-grandfather (the reader's great-grandfather is not Genghis
Khan or some other figure famous for extreme fertility). This seems to
contradict the IAP \red{axiom}.

The above objection mistakenly assumes there is something
special about being a person whose descendants eventually include all living
humans. But the theory makes no predictions about this status
being special in any way. The reader (or the reader's great-grandfather)
might have this status, but probably a significant proportion of all their
contemporaries also have the same status, so it is really nothing special.
To quote Socrates \cite{theaetetus}:
``When his [the philosopher's]
companions become lyric on the subject of great
families, and exclaim at the noble blood of one who can point to seven wealthy
ancestors, he thinks that such praise comes of a dim and limited vision, an
inability, through lack of education, to take a steady view of the whole, and
to calculate that every single man has countless hosts of ancestors, near and
remote, among whom are to be found, in every instance, rich men and beggars,
kings and slaves, Greeks and foreigners, by the thousand.''

A similar objection is that it seems inconceivable that someday all living
humans will visibly resemble one particular individual alive today
(for example, the author). But the identical ancestor point axiom
is about genealogical ancestry, and implies nothing about genetics.

\subsection{The ring species argument}

Certain species known as \emph{ring species} are partitioned into
subsets \red{$S_1,\ldots,S_n$}, with the property that
an o\redder{r}ganism in $S_i$ can only \red{have children in} $S_j$
\redder{when
$|i-j|\leq1$}. It seems plausible that
the descendants of some
member $x$ of $S_i$ might continue to proliferate indefinitely, and yet it
seems implausible for all the members of $S_j$ to eventually be
descendants of $x$, if
\redder{$|i-j|$ is large}. And this would seem to invalidate
the IAP axiom.

We can address this objection using an inductive argument. The apparent
implausibility hinges on \redder{$i$ being far from $j$}.
If $j=i\pm 1$
then the
likelihood of all of $S_j$ eventually being descended from $x$ is much
greater. But once enough time has passed that every member of $S_{i\pm 1}$
is descended from $x$, it is no longer necessary to prove all of $S_j$
is eventually descended from $x$: it becomes sufficient to prove all of $S_j$ is
eventually descended from some $y$ which inhabits $S_{i\pm 1}$ when all of
$S_{i\pm 1}$ is descended from $x$.
Since $i\pm 1$ and $j$ are \redder{closer than} $i$ and $j$ are
(for suitable choice of $\pm$), we can repeat this argument inductively
to reduce the problem to the case where $i=j\pm 1$.
\red{See also \cite{rohde2004modelling} for more discussion of how
geographic structure (such as rings) slows, but does not prevent,
reaching an identical ancestor point.}

\section{Summary and Conclusion}
\label{summarysection}

We are concerned with the problem of identifying species within the biosphere
using nothing besides the parent-child relationship graph and birthdates.
This seems to be
subjective, depending on many arbitrary choices, but in
Section \ref{objectivespeciessection} we show that if we require that species
satisfy the identical ancestor point axiom and the convexity axiom,
then, modulo a natural equivalence relation, this
subjectivity is reduced.

The identical ancestor point axiom, discussed in
Section \ref{speciesaxiomssection},
is the property that no member of the species can
both be ancestor to infinitely many members of that species and also
non-ancestor to infinitely many members of that species.
We gave five informal arguments that this is biologically plausible
as a law which species should satisfy (the mathematical results in the
paper are true whether or not species really do satisfy this axiom, or even
whether or not species exist at all).

\red{By a \emph{specieslike cluster} we mean a connected
set of organisms satisfying
the identical ancestor point axiom and the convexity axiom.
Some specieslike clusters are too small to be realistic species
candidates, so this motivates us to consider maximal specieslike
clusters. Unfortunately, there is no guarantee that every organism
inhabits a maximal specieslike cluster. This, in turn, motivates us
to ask what additional constraints could be imposed, not too
biologically implausible, in order to guarantee every organism
inhabits at least one maximal specieslike cluster subject to those
extra constraints. We give one example of such a constraint-set in
Theorem \ref{atomicspeciesexistencethm}.}

\section*{Data availability statement}

This is a theory paper, so there is no relevant experimental or otherwise
empirical data.

\section*{Compliance with ethical standards}

Disclosure of potential conflicts of interest: The author does not have
any conflicts of interest to disclose.

Research involving Human Participants and/or Animals and Informed consent:
The submission does
not involve experimental research involving human participants and/or animals.

\bibliographystyle{plain}
\bibliography{bio.bib}
\end{document}